\documentclass[10pt]{article}

\usepackage{amsmath}
\usepackage{times}
\usepackage{bm}
\usepackage{natbib}
\usepackage{amsfonts}

\usepackage{enumerate}
\usepackage{mathscinet}

\usepackage[plain,noend]{algorithm2e}

\usepackage{graphicx}
\usepackage[margin=1cm, textfont=small]{caption}

\usepackage[textwidth=1.5cm]{todonotes}
\usepackage{url}

\usepackage{booktabs,caption}
\usepackage[flushleft]{threeparttable}
\usepackage{setspace}

\allowdisplaybreaks

\newcommand{\abs}[1]{\left|#1\right|}
\newcommand{\pr}{{\mathrm \!pr}}
\newcommand{\I}{\mathcal I}
\newcommand{\IC}{\mathrm{\sc{PEN}}}
\DeclareMathOperator*{\argmin}{arg\,min}
\DeclareMathOperator*{\argmax}{arg\,max}
\DeclareMathOperator{\Tr}{Tr}

\newtheorem{theorem}{Theorem}

\newtheorem{definition}{Definition}

\newtheorem{proof}{Proof}

\begin{document}
\title{Seeded Binary Segmentation: A general methodology for fast and optimal change point detection}

\author{Solt Kov\'acs${}^{1}$, Housen Li${}^{2}$, Peter B\"uhlmann${}^{1}$, Axel Munk${}^{2}$\\
\vspace{0.1cm}\\
{\small${}^{1}$Seminar for Statistics, ETH Zurich, Switzerland}\\
{\small${}^{2}$Institute for Mathematical Stochastics, University of G\"ottingen, Germany}}

\date{February 2020}

\maketitle

\begin{abstract}
In recent years, there has been an increasing demand on efficient algorithms for large scale change point detection problems. To this end, we propose seeded binary segmentation, an approach relying on a deterministic construction of background intervals, called seeded intervals, in which single change points are searched. The final selection of change points based on the candidates from seeded intervals can be done in various ways, adapted to the problem at hand. Thus, seeded binary segmentation is easy to adapt to a wide range of change point detection problems, let that be univariate, multivariate or even high-dimensional. 

We consider the univariate Gaussian change in mean setup in detail. For this specific case we show that seeded binary segmentation leads to a near-linear time approach (i.e. linear up to a logarithmic factor) independent of the underlying number of change points. Furthermore, using appropriate selection methods, the methodology is shown to be asymptotically minimax optimal. While computationally more efficient, the finite sample  estimation performance remains competitive compared to state of the art procedures. Moreover, we illustrate the methodology for high-dimensional settings with an inverse covariance change point detection problem where our proposal leads to massive computational gains while still exhibiting good statistical performance.
\end{abstract}


\section{Introduction and Motivation}
\label{Introduction}
Change point detection refers to the problem of estimating the location of abrupt structural changes for data that are ordered e.g.~in time, space or by the genome sequence. One can distinguish between online (sequential) and offline (retrospective) change point detection. We are interested in the latter setup, where one has a set of ordered observations and the data in the segments between the change points (also called break points) are assumed to be coming from some common underlying distribution. The goal is to estimate both the location and the number of the change points. Applications cover areas in biology (e.g.~copy number variation in \citet{CBS}; or ion channels in \citet{Hotz_ionchannel}), finance \citep{Kim_finance}, environmental science (e.g.~climate data in \citet{Reeves_climatology}; or environmental monitoring systems in \citet{Londschien}) and many more.

There are two common approaches to tackle change point detection algorithmically: Modifications of dynamic programming (see e.g.~\cite{FKLW08,BoyKemLieMunWit09,Killick_etal}) aim to solve a global optimisation problem, which often lead to statistically efficient methods even in a multiscale manner \citep{Frick_etal};  Greedy procedures provide a heuristic to those (see e.g.~\citet{FJS19}) and are often based on binary segmentation \citep{Vostrikova}. Plain binary segmentation is usually fast and easy to adapt to various change point detection scenarios, but it is not optimal in terms of statistical estimation. \cite{Fryzlewicz_WBS} proposed wild binary segmentation and \cite{Baranowski} proposed the similar narrowest over threshold method which improve on statistical detection, but lose some of the computational efficiency of plain binary segmentation.
Both approaches draw a high number of random search intervals and in each the best split point is found greedily, leading to a set of candidate split points. The two methods differ in the way the final change point estimates are selected from these candidates.
The computational time for both essentially depends on the total length (and hence the number) of the drawn intervals. In frequent change point scenarios one should draw a higher number of random intervals both from a theoretical and practical perspective. The problems that arise are the following. In practice one may not know whether the number of drawn search intervals is sufficient and in the worst case, when essentially all possible segments are drawn, the total length of considered intervals will be cubic. Hence, as one considers all possible split points within the drawn intervals, going through all of them can already result in a computational cost that is cubic (e.g.~in univariate Gaussian change in mean setups) or possibly even worse if some complicated model fits are used for the investigated problem.

We propose seeded binary segmentation as a generic approach for fast and flexible change point detection in large scale problems. Our approach is similar to wild binary segmentation and the narrowest over threshold method in the sense that for many different search intervals the best split point for each is first determined greedily, in order to form a preliminary set of candidates out of which the final selection can be made. However, we propose a deterministic construction of search intervals which can be pre-computed. The resulting seeded intervals allow computational efficiency, i.e. lead to a total length for the search intervals that is linear up to a logarithmic factor (and thus we will refer to it as near-linear), independent of the underlying number of change points, while keeping the improvements in terms of statistical estimation compared to standard binary segmentation. Figure \ref{wbs_vs_sbs_computational_time} illustrates the statistical efficiency vs.~computational efficiency of seeded binary segmentation compared to wild binary segmentation for a univariate signal. It is remarkable that the estimation performance of the default option suggested for seeded binary segmentation is very close to what seems to be possible in this specific setup. Hence, our approach not only uses available resources more efficiently, but final results are also much more robust across the number of generated seeded intervals. 

\begin{figure}[h]
\vspace*{-30pt}
\includegraphics[width=1\textwidth]{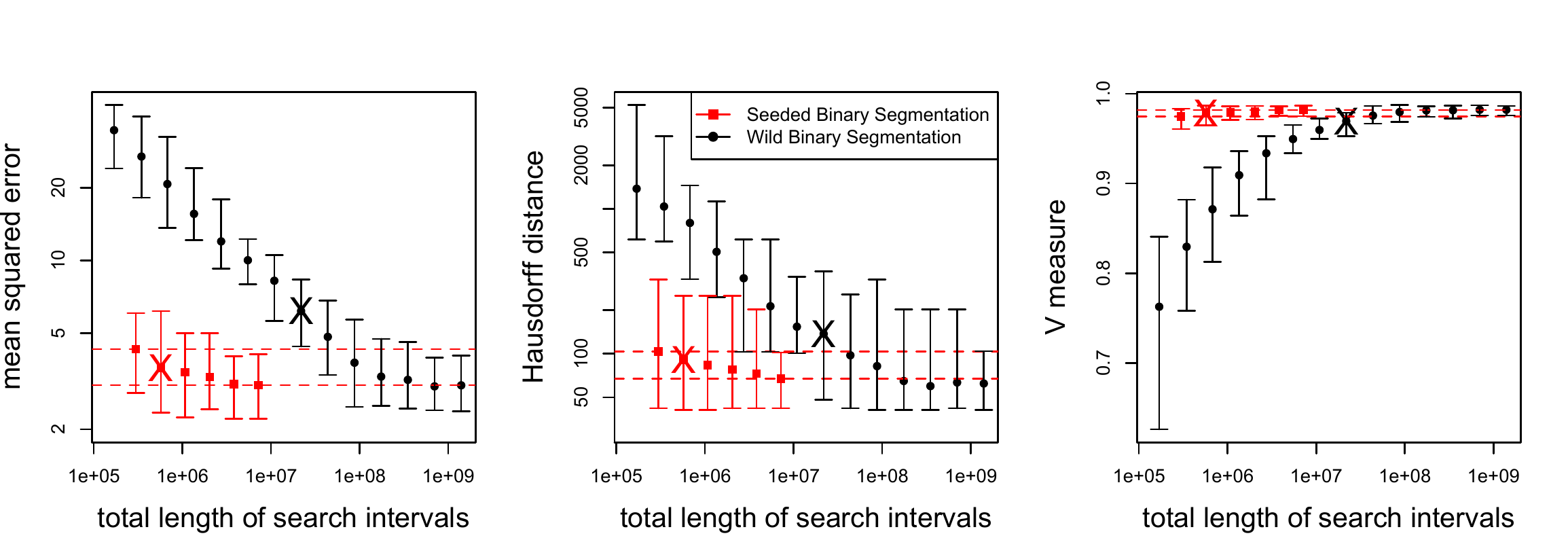}
\caption{Estimation performance on a five times repeated blocks signal \citep{DoJo94} with $5\cdot2048$ observations and $55$ change points. The dots and squares represent the average and the error bars the range of obtained values across $50$ simulations. Crosses correspond to the proposed default options for the methods.
The dashed horizontal lines are drawn as a visual aid. Seeded binary segmentation was combined with a greedy selection and the final models were chosen using the same information criterion as for wild binary segmentation. Note the logarithmic scales on both axes. More details on the blocks signal and the error measures can be found in section~\ref{Simulations}. }
\label{wbs_vs_sbs_computational_time}
\end{figure}

Two very useful applications for seeded binary segmentation in connection with ever larger data sets are in the following setups. First, applications with very long univariate Gaussian or similar low-dimensional signals with possibly hundreds of thousands of observations led to the need and thus emergence of computationally more efficient approaches. Based on dynamic programming, \cite{Killick_etal} proposed the pruned exact linear time method that is returning the optimal partitioning in linear time in frequent change point scenarios when pruning works well, but remains of quadratic complexity in the worst case (with a few change points only). \cite{Maidstone} proposed the functional pruning optimal partitioning method and demonstrated faster empirical performance compared to the pruned exact linear time method in many scenarios. \cite{Fryzlewicz_WBS2} recently proposed wild binary segmentation~2 for a univariate Gaussian change in mean setup, a hybrid between classical binary segmentation and wild binary segmentation aiming to improve the computational performance of the latter. In all these setups, seeded binary segmentation would lead to a competitive approach with guaranteed near-linear cost independent of the underlying number of change points as opposed to the previously mentioned methods (including even plain binary segmentation) which can all have a quadratic computational complexity in certain method-specific adverse scenarios. Note that dynamic programming would typically run fast in frequent change point scenarios while wild binary segmentation would run fast in the presence of only a few change points, but neither of them are guaranteed to be fast under all circumstances. 

However, a second, and to us even more appealing applicability of seeded binary segmentation is in setups where model fits are expensive. Several examples, especially in high-dimensional change point detection settings, are discussed at the beginning of section \ref{highdim_example}. The commonality among these setups is that they usually require computationally costly single fits and thus it is essential for practical feasibility to keep the total length of search intervals as small as possible. We will discuss this in more detail and using a concrete example in the last section.

The structure of the paper is as follows. In the next section we introduce seeded binary segmentation in more detail, background and with more discussion (e.g.~details regarding computational times). For the ease of explanation we do this in the classical setup of Gaussian variables with changing mean for which section~\ref{Theory} contains theoretical results regarding consistency and localisation error. Section~\ref{Simulations} provides simulation results which demonstrate competitive performance compared to other state-of-the-art estimation methods. As noted previously, the applicability of the seeded binary segmentation methodology is very wide and we illustrate this with a high-dimensional example in section \ref{highdim_example}.

\section{On the way to Seeded Binary Segmentation}
\label{SBS_Intro}
A well-studied problem, which we will analyse in detail as an illustrative example of the methodology, is the detection of shifts in mean of independent univariate Gaussian variables with variance that stays constant over time. Recent proposals for this setup are for example that of \cite{Frick_etal}, \cite{Fryzlewicz_WBS}, \cite{Du2016}, \cite{LMS16} and \cite{FJS19}, references therein also provide an overview. The understanding for this piecewise constant Gaussian setup can often be generalised to other scenarios, e.g.~for multivariate and dependent, for different parametric as well as non-parametric models, or even changes in piecewise linear structure. Hence, while focusing on this classical model for illustrative purposes, we would like to emphasise again the much wider potential applicability of our methodology (which we will come back to in the last section).

\subsection{The Gaussian setup and notation}
We consider the model
\begin{equation}\label{eq:model}
    X_t = f_t + \epsilon_t, \hspace{0.5cm} t = 1, \ldots, T,
\end{equation}
where $f_t$ is a fixed one-dimensional, piecewise constant signal with an unknown number $N$ of change points at unknown locations $\eta_1 < \dots < \eta_N \in \{2,\dots,T-1\}$. The noise $\epsilon_t$ is assumed to be independent and identically distributed Gaussians with zero mean and variance $\sigma^2$. Without loss of generality, the variance can be assumed to be 1. The test statistic used for detection by binary segmentation and related procedures is the CUSUM statistics \citep{Page54} defined for a split point $s$ within the segment $(\ell, r]$ as 
\begin{equation}
\label{CUSUM}
    T_{(\ell,r]}(s)= \sqrt{\frac{r-s}{n(s-\ell+1)}}\sum_{t=\ell+1}^s X_t - 
        \sqrt{\frac{s-\ell+1}{n(r-s)}}\sum_{t=s+1}^r X_t, 
\end{equation}
with integers $0\leq \ell<r\leq T$ and $n = r-\ell$. The CUSUM statistic is the generalised likelihood ratio test for a single change point at location $s$ in the interval $(\ell,r]$ against a constant signal. The value $|T_{(\ell,r]}(s)|$ will be referred to as gain, because the square of it gives the reduction in squared errors in the segment $(\ell,r]$ when splitting at point $s$ and fitting a separate mean on the corresponding left and right segments. The location of the maximal absolute CUSUM statistics 
\begin{equation*}
    \hat s_{(\ell,r]} = \argmax_{s\in{\{\ell+1, \dots, r-1\}}} |T_{(\ell,r]}(s)|
\end{equation*}
is the best split point and thus a candidate when dividing the segment into two parts. 
Note that in scenarios other than the Gaussian setup from equation (\ref{eq:model}) one would only need to change the CUSUM statistic to one that is adapted to the problem at hand. Whether one declares $\hat s_{(\ell,r]}$ to be a change point usually depends on the corresponding maximal gains 
$|T_{(\ell,r]}(\hat s_{(\ell,r]})|$.
We treat the final selection of change points as a separate issue and will discuss specific ways to do so later on.

\subsection{Related methods based on binary segmentation}
Binary segmentation \citep{Vostrikova} first finds the best split $\hat s_{(0,T]}$ on all data and if $|T_{(0,T]}(\hat s_{(0,T]})|$ is above a pre-defined threshold, then searches for other split points on the left and right segments $(0,\hat s_{(0,T]}]$ and $(\hat s_{(0,T]},T]$. The same iterative search is continued as long as the gains are still above the pre-defined threshold and there are sufficiently many observations in the resulting segments.

The wild binary segmentation method \citep{Fryzlewicz_WBS} and the narrowest over threshold method \citep{Baranowski} both generate $M$ random intervals $I_m = (\ell_m, r_m]$, $m=1, \ldots, M$, where the starting and end points are drawn independently, uniformly and with replacement from $\{0, 1, \ldots, T \}$, possibly subject to some minimal length constraint to be able to evaluate the chosen test statistic. One determines candidate splits and corresponding maximal gains for each of the random intervals via the CUSUM statistics defined in equation (\ref{CUSUM}). The two methods are slightly different beyond this point, in the way to select the final change point estimates.

Wild binary segmentation orders the candidates according to the respective maximal gains (i.e. values of the test statistic) and selects greedily, i.e. the split point with highest value is the first change point estimate. Then all intervals are eliminated from the list which contain this found change point. From the remaining candidates again the one with maximal gain is selected, followed by the elimination of all intervals that contain this second one. The procedure is repeated as long as the split with highest gain remains above a certain pre-defined threshold. A suitable threshold can be derived from asymptotic theory for example. Out of several possible segmentations corresponding to different thresholds, one can also select the one optimising a chosen information criterion (e.g.~proposed by \cite{mBIC_Zhang}).

The narrowest over threshold method selects change points amongst the candidates that have a maximal gain above a pre-defined threshold. Amongst these candidates the one having the shortest (i.e. narrowest) background interval is selected first. Then all intervals that contain this first estimate are eliminated and repeatedly from the left over ones above the pre-defined threshold the next one with the currently shortest background is selected. The procedure is repeated as long as no more maximal gain is above the pre-defined threshold. It is also possible to repeat this procedure with differing thresholds and then take the one optimising some information criterion. Note that the narrowest over threshold way of selecting the final estimates makes it applicable in more general scenarios, where the model is e.g.~piecewise linear.

Elimination of the intervals in each step for both methods is necessary in order to avoid detecting the same change point at slightly different locations for several different background intervals. The two methods allow for detection in a larger class of signals compared to binary segmentation. The intuition behind this is that some of the random intervals will contain only a single change point sufficiently well separated from the boundaries and thus serve as good local backgrounds in which change points are easily detectable. As a comparison, binary segmentation is particularly bad in signals that exhibit many up and down jumps that cancel each other if a background contains several change points. 

The main drawback of wild binary segmentation and the narrowest over threshold method is the increased computational cost. This grows linearly with the total length of used random intervals (and with $O(T)$ expected length for each random interval). In frequent change point scenarios one should draw a higher number of random intervals in order to ensure that there are still sufficiently many intervals with a single change point in them. In the worst case, when drawing essentially all possible intervals, the complexity is cubic, which is clearly not satisfactory for large scale change point detection problems. As mentioned previously, it is also an issue that one might not know how many change points to expect and thus how many intervals to draw. \cite{Fryzlewicz_WBS2} summarises these issues as the lack of computational adaptivity of wild binary segmentation and proposes a modification, called wild binary segmentation~2, which is a hybrid between wild and classical binary segmentation in the following sense. In order to speed up, he suggests to draw a smaller number $M' << M$ of random intervals, localise the first candidate greedily, and in the resulting left and right segments draw again $M'$ random intervals, localise again, and repeat iteratively. From the resulting preliminary list of candidates the final estimates are again determined based on certain maximal gains. However, this proposal is fairly specific to the univariate Gaussian setup and might fail when adapting to high-dimensional or other more complicated scenarios such as a model that is piecewise linear. An interesting two-stage procedure also targeting computational issues for the univariate Gaussian case was proposed by \cite{intelligent_sampling}. They obtain pilot estimates of change points on a sparse subsample and then refine these estimates in a second step in the respective neighbourhoods using dense sampling. This can reduce the computational cost to sub-linear, and is thus useful for extremely long time series. However, note that computational issues in high-dimensional setups primarily arise from the fact that data sets are too `wide' and not because of being `long'. Moreover, there is a price to pay in terms of detectable changes using this procedure, namely, being statistically suboptimal.

Next, we propose our conceptually simple seeded binary segmentation method that has a number of advantages. To name a few, it leads to a near-linear time algorithm under all circumstances, independent of the number of change points, it simplifies theoretical analysis and has an applicability wider than the Gaussian scenario.

\subsection{Seeded Binary Segmentation}
The question to start with regarding the randomly generated intervals above is, whether all of them are necessary? When the starting and end points are drawn uniformly at random, there will be many intervals with length of the order of the original series (see bottom of Fig.~\ref{interval_visualization}). While these are one of the main drivers of computational time, if there are many change points, such long intervals are often not really useful as they contain several change points. One should rather focus on shorter intervals. As an example, the total length of random intervals in the bottom of Fig.~\ref{interval_visualization} is approximately the same as for the seeded intervals shown on the top of Fig.~\ref{interval_visualization}. While the shown random intervals only contain a few short intervals and also not covering short intervals uniformly, the seeded intervals achieve this much better. Formally, we propose the following deterministic construction guaranteeing a good coverage of single change points.

\begin{definition}[Seeded intervals]
\label{def:seeded_intervals}
Let $a\in[1/2,1)$ denote a given decay parameter. For $1 \leq k \leq \lceil\log_{1/a}(T)\rceil$ (i.e.~logarithm with base $1/a$) define the $k$-th layer as the collection of $n_k$ intervals of initial length $l_k$ that are evenly shifted by the deterministic shift $s_k$ as follows:
\begin{equation*}
    \I_k = \bigcup\limits_{i=1}^{n_k} \{ (\lfloor (i-1)s_k \rfloor,
        \lceil(i-1)s_k + l_k \rceil ] \},  
\end{equation*}
where
$n_k = 2\lceil (1/a)^{k-1} \rceil - 1$,  
$l_k  = Ta^{k-1}$ and
$s_k  = (T-l_k)/(n_k-1)$.
The overall collection of seeded intervals is
\begin{equation*}
    \I = \bigcup\limits_{k=1}^{\lceil\log_{1/a}(T)\rceil} \I_k.
\end{equation*}
\end{definition}
Some technical remarks are due first. Due to rounding to integer grid values, the effective lengths and shifts are in practice slightly different from $l_k$ and $s_k$. Moreover, some intervals might appear multiple times on small scales (depending on choice of $a$). Such redundant ones can be eliminated. 
Using some rough bounds, one easily obtains a near-linear bound for the total length of all the intervals contained in $\I$: $\sum_{k=1}^{\lceil\log_{1/a}(T)\rceil} \sum_{i=1}^{n_k} \lceil l_k \rceil \leq 6T\lceil\log_{1/a}(T)\rceil$.
Thus, evaluating the CUSUM statistics for all of these intervals also has a computational cost of the order $O(T\log_{1/a}(T))$. Note that one can simply stop at an earlier layer than the $\lceil\log_{1/a}(T)\rceil$-th one if a certain minimal segment length is desirable (e.g.~in high-dimensional scenarios) or when prior knowledge on the true minimal segment length is available. 

While the lengths of intervals decay rapidly over the layers, the number of intervals is rapidly increasing as the scales become smaller. In particular, at the lowest layer the total number of intervals is of the same order as the total number of observations, i.e. for $k = \lceil\log_{1/a}(T)\rceil$, $\abs{\I_k}$ is of the order $O(T)$. This allows a good coverage even in frequent change point scenarios. Hence, in some sense, the number of intervals is automatically adapted to each scale. Figure~\ref{interval_visualization} shows the first few layers of two collections of seeded intervals as well as some random intervals for a comparison. The key idea is best understood for $a=1/2$ (top of Fig.~\ref{interval_visualization}). Here one obtains an overlapping dyadic partition. Each interval from the previous layer is split into a left, a right and an overlapping middle interval, each of half the size of the previous layer. The middle plot of Fig.~\ref{interval_visualization} shows seeded intervals of slower decay, namely $1/2^{1/2}$. Here the solid lines in layers $\I_1, \I_3$ and $\I_5$ show intervals that were also included in the interval system for $a=1/2$ (top), while the dashed lines are additionally obtained intervals. These intervals on the intermediate layers $\I_2$ and $\I_4$ are still evenly spread. In particular, a sufficiently big overlap amongst the consecutive intervals within the same layer is guaranteed. This is key, because one obtains for each change point background intervals with only that single change point in there and such that the change point is not too close to the boundaries, which facilitates detection. We recommend using $a=1/2^{1/2}$ as a decay in practice, but as indicated above, adding intermediate layers is very easy when trying to achieve better finite sample performance in very challenging scenarios. 

\begin{figure}[t]
\includegraphics[width=0.95\textwidth]{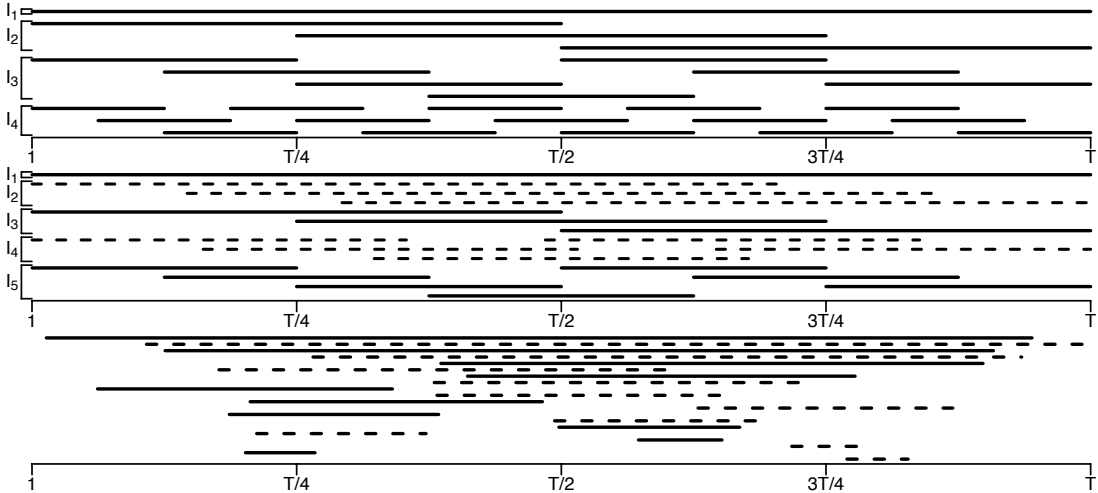}
\caption{The first four (with $a=1/2$, top) and first five layers (with $a=(1/2)^{1/2}$, middle) of seeded intervals, as well as $20$ random intervals ordered according to length (bottom). The division into dashed and solid lines has no particular meaning, it is only for visual aid.}
\label{interval_visualization}
\end{figure}

Several related interval constructions were proposed by Guenther Walther together with various co-authors. First proposed for the detection of spatial clusters \citep{Walther}, as well as density estimation \citep{Rufibach_Walther}, these proposals were then gradually adapted towards more classical change point detection (or fairly related problems) by \cite{Rivera_Walther, Chan_Walther2013, Chan_Walther2015}. The difference between these various intervals and seeded intervals mainly lies in the specific choices of intervals for each scale. \cite{Pein_etal} also considered dyadic partitions, but as a major difference, without an overlap, moreover, with no flexibility in decay. Recently, yet another interval system was proposed by \cite{Chan_Chen2017} in the context of bottom up approaches. However, these specific interval construction proposals appeared as tricks and side remarks rather than pronouncing the potential of a stand-alone flexible methodology. We believe our construction is more intuitive, in particular showing the trade-offs between computational time and statistical precision.

Algorithm \ref{alg:seeded_binary_segmentation} describes our seeded binary segmentation. The decay parameter $a$ governs the trade-off between computational time and expected estimation performance. Limiting the minimal number of observations $m$ is reasonable e.g.~in high-dimensional scenarios or if prior information on the minimal true segment lengths is available. Choosing a larger value not only speeds up computations, but also reduces the risk of selecting false positives coming from very small intervals during the final selection. We find it important to disentangle the construction of search intervals from the final selection as various options with advantages and disadvantages arise. Some possible selection methods are discussed in section \ref{selection_method}. 

\begin{algorithm}[!h]
\vspace{-0.3cm}
\caption{Seeded Binary Segmentation}
\label{alg:seeded_binary_segmentation}
\begin{tabbing}
   \enspace Require \\
   \qquad a decay parameter $a\in[1/2,1)$, 
   a minimal segment length $m \ge 2$ and 
   a selection method. \\
   \enspace Create a collection $\I$ of seeded intervals with decay $a$ from Definition \ref{def:seeded_intervals}, which cover $\ge m$ observations.  \\
   \enspace For $i=1$ to $|\I|$ \\
   \qquad Take the $i$-th interval in $\I$ and denote its boundaries with $\ell$ and $r$. \\
   \qquad Calculate the CUSUM statistics $T_{(\ell,r]}(s)$ from equation (\ref{CUSUM}) for $s = \ell,\ldots, r-1$.\\
\enspace Apply the chosen selection method to         $T_{(\ell,r]}(\cdot), (\ell,r] \in \I$ to output the final change point estimates.
\end{tabbing}
\end{algorithm}

Overall, the key components for our seeded interval construction are 1) a sufficiently fast decay such that the total length of all intervals is near-linear, 2) choosing a larger number of intervals on the small scales and thus suitably adapting for the possibility that there can be many more small segments than large ones, and last, but not least 3) guaranteeing a sufficient overlap amongst intervals on the same layer. The choice of $a$ can be interpreted as a trade-off between computational time and expected performance. 

\subsection{Selection methods}
\label{selection_method}
While mentioned earlier, we recall some possible selection methods here. Greedy selection takes the location of the maximal CUSUM value ($\max_{(\ell,r]\in \I} |T_{(\ell,r]}(\hat s_{(\ell,r]})|$) as the change point candidate, then eliminates from $\I$ all intervals that contain this candidate, and repeats these two steps as long as there are some intervals left. From the list of these candidates all above a certain threshold (typically some constant times $\log^{1/2} T$) are kept. Alternatively, one can also investigate various thresholds, each leading to a different segmentation, and at the end select the segmentation that optimises some information criterion. 

The narrowest over threshold selection takes all intervals $(\ell,r]\in \I$ for which $T_{(\ell,r]}(\hat s_{(\ell,r]})$ is over a pre-defined threshold (typically some constant times $\log^{1/2} T$). Out of these it takes the shortest (narrowest) interval $(\ell^*,r^*]$, i.e. the one for which $r^*-\ell^*\leq r-\ell$ for all intervals $(\ell,r]$ with maximal CUSUM value above the threshold. Then it eliminates all intervals that contain the found change point candidate $\hat s_{(\ell^*,r^*]}$, and repeats these two steps as long as there are no more intervals left with a maximal CUSUM statistics above the pre-defined threshold. While raising the threshold tends to result in fewer change points, due to the narrowest selection followed by the elimination step, it could also happen that one obtains more change points. Hence, in some scenarios this selection method can be somewhat unstable, as even similar thresholds can lead to segmentations differing both with respect to number and location of the boundaries. It is also possible to evaluate a set of different segmentations arising from different thresholds in terms of some information criterion and then pick the overall best.

Various other options, e.g.~the steepest drop to low levels proposed by \cite{Fryzlewicz_WBS2} or selection thresholds guaranteeing false positive control either via theoretical approximation results \citep{FJS19} or via simulations (facilitated by the favourable scaling of the total length of seeded intervals of $O(T \log T)$) could be adapted as well. Besides the choice of $a$, the overall computational costs also depend on how complicated the final model selection is. Greedy selection for example is very fast, while full solution paths (using all possible thresholds) of the narrowest over threshold method can be slow. Whether this becomes an issue depends on the problem at hand. In high-dimensional setups the expensive model fits would typically dominate the cost of model selection anyway. To make this more formal, consider the following theorems.

\begin{theorem}\label{th:fixThd}
Assume model~\eqref{eq:model} and that in seeded binary segmentation (i.e., Algorithm~\ref{alg:seeded_binary_segmentation}) we perform either the greedy selection or the narrowest over threshold selection, with an arbitrarily fixed single threshold. Then the computational complexity of seeded binary segmentation is $O(T \log T)$  and the memory complexity is $O(T)$.
\end{theorem}

\begin{proof}
The algorithm consists of two steps: the calculation of the CUSUM statistics from \eqref{CUSUM} for all seeded intervals and the subsequent model selection. For the first step, the cost of calculating CUSUM statistics on all seeded intervals is proportional to their total length, i.e., $O(T \log T)$.

In the second step, we only consider the candidate split points from seeded intervals with maximal CUSUM values above the given threshold, and sort them either according to their maximal gains (for the greedy selection) or their lengths (for the narrowest over threshold selection). Note that the number of intervals contained in a collection of seeded intervals is of the order $O(T)$. Hence, the sorting procedure involves $O(T \log T)$ computations in the worst case. Then we check for each seeded interval in the sorted list whether it contains any previously found change point. Such a check can be done in $O(\log T)$ computations as follows. We store all previously found change points in a balanced binary tree structure \citep[e.g., the AVL tree by][]{AVL62}, where both the search and the insertion of new elements have the worst case complexity $O(\log T)$, and the memory complexity is $O(T)$ (given the maximum of $T$ total points that could potentially be declared as change points and thus need to be stored in the AVL tree). Thus, one needs $O(\log T)$ time for checking whether the candidate split is suitable (i.e.~the corresponding search interval does not contain any previously found change point, for which finding the nearest neighbouring points in the AVL tree to the currently considered candidate is sufficient), and $O(\log T)$ time for updating the AVL tree supposing that the candidate is suitable and thus that it is declared/detected as a new change point. Alternatively, we can utilise the structure of seeded intervals by noticing that each potential split point $s \in \{1,\ldots,T\}$ is contained in at most $O(\log T)$ intervals (as in each of the $O(\log T)$ scales there is a constant number of intervals containing $s$). This allows us to mark all of the intervals that contain the found change point with a cost of $O(\log T)$ such that these intervals can then be skipped in consecutive steps when adding the new change points. Overall, both approaches (AVL trees and the specific structure of seeded intervals) lead to a worst case complexity $O(T\log T)$ for the second step as we go through at most $O(T)$ intervals (the total number of intervals in a collection of seeded intervals) and perform at most $O(\log T)$ computations at each (checking whether the interval is suitable and an update of the AVL tree). Therefore, the total computational complexity is $O(T \log T)$ and it does not depend on the underlying number of true change points. The total memory complexity is composed of the storage of information (maximal gain, best split, as well as start and end point for each search interval) plus the AVL tree, all of which requires at most $O(T)$ memory. 
\end{proof}

\begin{theorem}\label{th:infoSel}
Assume model~\eqref{eq:model} and that in seeded binary segmentation (i.e., Algorithm~\ref{alg:seeded_binary_segmentation}) we determine a full solution path (i.e., calculate segmentations with all thresholds corresponding to maximal gains in any of the seeded intervals) and then select the final model out of the ones in the solution path based on an information criterion. Specifically, consider an information criterion with a penalty $\IC(\cdot)$ of the form
\begin{equation}\label{eq:ic}
    \tilde\eta_1<\cdots<\tilde\eta_m \quad
\mapsto\quad
\sum_{i = 1}^m \sum_{t \in (\tilde\eta_i,\tilde\eta_{i+1}]} (\bar X_{(\tilde\eta_i,\tilde\eta_{i+1}]} - X_t)^2 + \IC(\{\tilde\eta_1,\ldots,\tilde\eta_m\})
\end{equation}
with $\bar X_{(i,j]} = \sum_t X_t / (j-i)$, $\tilde\eta_0 = 0$ and $\tilde\eta_{m+1} = T$. Assume further that given $\IC(S)$ the computation of $\IC(S \cup \{\tilde\eta\})$ is $O(1)$, for every finite set of split points $S$ and any arbitrary additional split point $\tilde\eta$. Then seeded binary segmentation has the memory complexity $O(T)$ and the computational complexity of 
$$
\begin{cases}
    O(T \log T) &\quad \text{if performing the greedy selection},\\
    O(T^2 \log T) &\quad \text{if performing the narrowest over threshold selection};
\end{cases}
$$
when the segmentation with the best value of the information criterion \eqref{eq:ic} in the solution path is chosen.
\end{theorem}

We note that the information criterion \eqref{eq:ic} in Theorem~\ref{th:infoSel} is fairly general, including e.g., the Akaike information criterion, the Bayes information criterion and its modifications \citep{mBIC_Zhang} as special cases. The upper bound for the narrowest over threshold selection in Theorem~\ref{th:infoSel} is worse than that for the greedy selection, because, as mentioned previously, for the narrowest over threshold selection even similar thresholds sometimes lead to solutions differing both with respect to number and location of the change point estimates. Note that for the cases when close-by thresholds produce similar models, one could reuse parts of previous segmentations and thus improve on the worst case bound. We do not know whether any smart updating (potentially using the specific deterministic structure of seeded intervals) could lead to some provably better complexity than the crude bound stated. For this one would need to be able say something about the instability in the solution path (i.e., bound the number of cases when similar thresholds produce very different segmentations).

From a practical perspective, thresholds $\kappa = C \log^{1/2} T$ are of interests for some constant $C$. Thus, instead of investigating all possible thresholds, picking some of this specific order could be sufficient to obtain good empirical results while maintaining fairly fast computations, i.e., $O(T \log T)$.   
Alternatively, one could use a reasonable threshold (maybe even a few different ones) and then create a `solution path' by deleting the found change points in a bottom up fashion based on local values of the test statistic (i.e., at each step one deletes the change point which has the smallest CUSUM value in the segment determined by its closest left and right neighbours, similar to the reverse segmentation proposal of \cite{Chan_Chen2017}) to obtain a hierarchy of nested models for which the evaluation of the information is cheap (similar to greedy selection). Yet another option is to consider fewer layers in the collection of seeded intervals, leading to much fewer intervals in total (as most intervals are at the lowest layers covering small scales) and thus also much fewer possible thresholds to be used for the solution path. While the overall cost of the selection procedure could be reduced massively, the speedup would come at a price, namely, the risk of not detecting true change points with a small spacing in between. Last, but not least, computations in parallel are easy to implement both for the calculation of the test statistic in the various search intervals, as well as the model selection using different thresholds (in the narrowest over threshold selection).

\begin{proof}
In order to apply the information criterion~\eqref{eq:ic}, we need to compute the full solution path, namely, the solutions for all possible choices of thresholds. Consider first the greedy selection, which has the nice property that the obtained solutions are nested, i.e.~adding another split point (equivalent to lowering the threshold) splits one of the existing intervals into two parts. Hence, to obtain the full solution path, one can just let the model selection run as in Theorem~\ref{th:fixThd}, but with a threshold set to~$0$, leading to at most $O(T\log T)$ computations. Further, the penalty $\IC(\cdot)$ can be evaluated in an incremental way, i.e., each time when adding a new change point, which by assumption needs $O(1)$ computations per update. For the least squares part in~\eqref{eq:ic}, it is also possible to update incrementally with $O(\log T)$ computations at each step. More precisely, we have to find the nearest left and right point (out of all previously found change points) to the newly added change point. Once the neighbours are found, one can just update the least squares of the specific segment that was freshly divided in criterion~\eqref{eq:ic} in $O(1)$ time by utilising cumulative sum type of information (all of which can be pre-computed once in the beginning in $O(T)$ time, and stored in $O(T)$ space). Such searches of nearest neighbours can be done in $O(\log T)$ time, if we build and iteratively update a balanced binary tree \citep[e.g., the AVL tree by][]{AVL62} for the found change points as in the proof of Theorem~\ref{th:fixThd}. Note that such balanced trees can be updated incrementally also in $O(\log T)$ time. Given that we go through at most $O(T)$ intervals and perform at most $O(\log T)$ computations at each, we thus obtain an upper bound on the worst case computational complexity of $O(T \log T)$ for the greedy selection.  

Consider now the narrowest over threshold selection. Note that there are at most $O(T)$ seeded search intervals and thus thresholds which can potentially lead to different solutions. For each threshold, we can compute the solution in $O(T \log T)$ time by Theorem~\ref{th:fixThd} and evaluate the information criterion~\eqref{eq:ic} in $O(T \log T)$ time (i.e., in an incremental way similarly to the greedy selection considered above). Thus, in total, we have an upper bound on the worst case complexity being $O(T^2 \log T)$.

Note that only the value of the information criterion \eqref{eq:ic} needs to be stored for every threshold that leads to a different solution, requiring $O(T)$ memory. The final solution can be calculated again with $O(T\log T)$ time once the optimal threshold is selected. Thus, the memory complexity is $O(T)$ for both selection methods. Of course, if one wants to record all possible segmentations from the solution path of the narrowest over threshold selection, the memory complexity could be $O(T^2)$. 
\end{proof}

\section{Statistical theory}
\label{Theory}
The introduction of seeded intervals (in Definition~\ref{def:seeded_intervals}) not only accelerates the computation, but also simplifies the statistical analysis of the resulting segmentation method, i.e., the seeded binary segmentation (see Algorithm~\ref{alg:seeded_binary_segmentation}). Consistency results depend on the assumed model as well as the selection method. As an illustration, we consider again the Gaussian change in mean model \eqref{eq:model} for which wild binary segmentation \citep{Fryzlewicz_WBS} and the narrowest over threshold method \citep{Baranowski} were claimed to be minimax optimal. However, the proof of the former seems to contain errors, e.g., some bounds for the single change point case were wrongly applied to the multiple change point cases. This has been noted by \cite{WYR18} in the (sub)Gaussian setup. Contrary to what is claimed by \cite{Fryzlewicz_WBS}, the greedy selection utilised in wild binary segmentation does not lead to optimal detection (and the same is also true for wild binary segmentation~2 of \cite{Fryzlewicz_WBS2} as the therein utilised steepest drop to low levels selection is also a certain greedy selection method). Put differently, without additional assumptions on the generated intervals or modification of the selection method itself, consistency seems only achievable for a smaller class of signals than claimed originally. Albeit in a nonparametric setup, \cite{Kolmogorov_Smirnov_CP} mention the errors of the original proofs of \cite{Fryzlewicz_WBS} as well, along some discussion on an additional assumption that they impose under which nearly optimal results can be attained (see their second remark after Theorem~2 on page~8). However, their assumption is a restriction on the generated search intervals requiring knowledge of the minimal segment lengths. Depending on the investigated data, such a background knowledge may or may not be a realistic assumption. Note that from a practical perspective greedy selection still offers advantages (e.g. easy applicability and implementation in various scenarios, segmentations are nested and stable across thresholds, etc.) and in terms of empirical performance it is often at least as good as the narrowest over threshold method for example (see the simulations in section~\ref{Simulations}). As seeded binary segmentation would suffer similarly if the greedy selection method is chosen, we show the strong consistency using the narrowest over threshold selection method \citep{Baranowski} as an illustration.

\begin{theorem}
\label{th:sbsNOT}
Assume model~\eqref{eq:model} and let $\delta_* = \min_{i = 1,\ldots,N}\delta_i$ with $\delta_i = \abs{f_{\eta_i + 1} - f_{\eta_i}}$, and $\lambda = \min_{i = 0, \ldots, N} \abs{\eta_{i+1} - \eta_i}$ with $\eta_0 = 0$ and $\eta_{N+1} = T$. Assume also that 
\begin{equation}\label{eq:ass}
\delta_*^2\lambda \ge C_0\log T \qquad \text{ for some constant }C_0. 
\end{equation}
Let $\hat N$ and $\hat\eta_1 < \cdots < \hat \eta_{\hat N}$ denote respectively the number and locations of estimated change points by the seeded binary segmentation in Algorithm~\ref{alg:seeded_binary_segmentation} with the narrowest over threshold as the selection method. Then there exist constants $C_1$, $C_2$, independent of $T$ and $a$ (in Definition~\ref{def:seeded_intervals}), such that, given the threshold for the selection method 
$\kappa = C_1 \log^{1/2} T$, as $T\to \infty$, 
\begin{equation}\label{eq:c}
\pr\left\{\hat N = N,\, \max_{i = 1, \ldots, N} \delta_i^2\abs{\hat \eta_i - \eta_i} \le C_2 \log T\right\} \to 1\,.
\end{equation}
\end{theorem}

The signal $f$ as well as $\delta_i$, $\delta_*$ and $\lambda$ is allowed to depend on $T.$ The assumption for detection in \eqref{eq:ass} and the localisation rates $\delta_i^{-2}\log T$ in \eqref{eq:c} are minimax optimal up to constants~\citep[see e.g.][]{Chan_Walther2013,Frick_etal,Chan_Chen2017,WYR18,LGM19}. That is, no method can detect all $N$ change points with weaker condition than \eqref{eq:ass}, and no method can achieve a better rate than $\delta_i^{-2}\log T$, in the asymptotic sense. 

\begin{proof}
We consider interval $(\eta_i - \lambda, \eta_i+\lambda]$ for each change point $\eta_i$. By construction of seeded intervals in Definition~\ref{def:seeded_intervals}, we can find a seeded interval $(c_i - r_i, c_i + r_i]$ such that $(c_i - r_i, c_i + r_i] \subseteq (\eta_i - \lambda, \eta_i+\lambda]$, $r_i \ge a^2 \lambda$ and $\abs{c_i - \eta_i} \le r_i (1 + a^2)/2$. Then the problem almost reduces to the detection of change point on the interval $(c_i - r_i, c_i + r_i]$. The rest of the proof follows nearly the same as the proof of Theorem 1 in \citet{Baranowski}, and is thus omitted.  
\end{proof}

We can further establish the consistency of the seeded binary segmentation using certain data adaptive selection rule (i.e., strengthened Schwarz information criterion).  

\begin{theorem}
Assume notation from Theorem~\ref{th:sbsNOT}, and that $\lambda \ge \log^{\theta_0} T$, and  $N$, $\delta_j$, $j=1,\ldots,N$ lie in $(c_1, c_2)$ for some constants $\theta_0 > 1$ and $0<c_1<c_2< \infty$. Let $\hat N$ and $\hat \eta_1 < \cdots < \eta_{\hat N}$ be respectively the number and locations of estimated change points by the seeded binary segmentation in Algorithm~\ref{alg:seeded_binary_segmentation} with the narrowest over threshold, the threshold of which is determined by the strengthened Schwarz information criterion using $\theta\in (1,\theta_0)$. Then there is a constant $C$ independent of $T$  and $a$ (in Definition~\ref{def:seeded_intervals}) such that, as $T \to \infty$, 
$$
\pr\left\{\hat N = N,\, \max_{i = 1, \ldots, N} \delta_i^2\abs{\hat \eta_i - \eta_i} \le C \log T\right\} \to 1\,.
$$
\end{theorem}

\begin{proof}
It follows exactly the same way as the proof of Theorem 3 in \cite{Baranowski}.
\end{proof}

In a nutshell, seeded binary segmentation combined with the narrowest over threshold selection has the same statistical guarantee as the original narrowest over threshold method of \cite{Baranowski}, but requires far less computations, with the worst case computational complexity being nearly linear (for a single threshold based selection), see~Theorem~\ref{th:fixThd}.

\section{Empirical results}
\label{Simulations}
We consider the same simulations as presented by \citet{Fryzlewicz_WBS} in connection with wild binary segmentation. For a description of the signals see Appendix B of \citet{Fryzlewicz_WBS}. We used implementations of the wild binary segmentation and the narrowest over threshold method from the \textsf{R} packages \textbf{wbs} and \textbf{not} from CRAN, respectively, with a small bug fixed in both packages concerning the generation of random intervals. Moreover, for our proposed seeded binary segmentation methodology we slightly modified the code in \textbf{wbs} to allow for seeded intervals as inputs. The code is available at \url{https://github.com/kovacssolt/SeedBinSeg}. Arbitrary intervals as inputs were already supported by the \textbf{not} package. 
Note however, that the latter implementation can be slow as all segmentations in the solution path resulting from all possible thresholds are recorded, already leading to a quadratic number of candidates being saved (see Theorem \ref{th:infoSel} and corresponding discussions for details).

Table \ref{tab:simulation_results} shows results based on averaged values over $100$ simulations for each setup. Reported are the mean squared errors (MSE), Hausdorff distances, V measures \citep[more precisely, viewing the change point segmentation as a clustering problem by identifying each constant part as a cluster]{RoHi07}, errors in the estimates for the number of change points as well as the total length of the constructed search intervals. Parentheses show the corresponding standard deviations. For our seeded binary segmentation methodology we chose $m=2$ and $a=1/2^{1/2}$. Exceptions are the rows marked with a star, where the decay was chosen to be $a=1/2^{1/8}$, i.e. an increase in the total length of search intervals by a factor of roughly four. For both the wild binary segmentation and the narrowest over threshold methods we chose the recommended $5000$ random intervals. Moreover, all methods selected the final model using the same information criterion. 

Overall, seeded binary segmentation with greedy selection (g-SeedBS) performs similar to wild binary segmentation (WBS), while seeded binary segmentation with the narrowest over threshold selection (n-SeedBS) is similar to the original narrowest over threshold method (NOT). Most importantly, using seeded intervals we achieve competitive performance with total lengths of search intervals that are orders of magnitudes smaller than using the random intervals. In these scenarios both selection types perform similarly. We have the experience that increasing the noise level would lead to better results using the greedy selection. The performance of our default methodology (with $a=1/2^{1/2}$) is slightly worse for the blocks signal compared to random intervals based counterparts. However, when choosing  $a=1/2^{1/8}$, and thus accepting an increase in the total length of the search intervals by about a factor four, performances are very close again. Note that the total length of search intervals is still only a tenth of that of random intervals. For a better understanding, the simulations for the blocks signal are depicted in Fig.~\ref{fig:estimation_visualization}. In the top row (with $a=1/2^{1/2}$) one can see points lining up vertically when considering Hausdorff distances. This line is around $40$, and indicates that one of the true change points (lying to this distance to its closest neighbour) is missed. However, when increasing computational effort ($a=1/2^{1/8}$, bottom row), this vertical line basically disappears and for all three error measures many points line up exactly on the identity line. This indicates that found change points using seeded binary segmentation with greedy selection and wild binary segmentation exactly agree in many simulation runs.

\begin{table*}[ht]
\caption{Simulation results}
\label{tab:simulation_results}
\begin{threeparttable}
\begin{tabular}{llccccr@{}c}
signal & method	&	MSE	&	Hausdorff dist.	&	V measure	&	$N - \hat N$	&	\multicolumn{2}{c}{tot. length} \\
\hline
blocks 
&g-SeedBS	&	2.922 (1.077)	&	43.150 (31.009)	&	 0.970 (0.013)	&	-0.610 (0.803)  &    95.3   & (0)       \\
&g-SeedBS*	&	2.685 (0.828)	&	32.690 (21.037)	&	 0.973 (0.009)	&	-0.520 (0.577)	&	329.7   & (0)       \\
&WBS	    &	2.627 (0.854)	&	31.520 (23.262)	&	 0.973 (0.011)	&	-0.500 (0.577)	&   3419.1  &\ (35.7)    \\
&n-SeedBS	&	2.942 (1.002)	&	42.630 (28.690)	&	 0.970 (0.013)	&	-0.690 (0.787)	&   95.3    & (0)       \\
&NOT		&	2.539 (0.874)	&	31.700 (22.778)	&	 0.974 (0.011)	&	-0.590 (0.588)	&   3422.3  &\ (32.6)    \\
&n-SeedBS* 	&	2.675 (0.782)	&	33.790 (21.056)	&	 0.973 (0.010)	&	-0.600 (0.532)	&   329.7   & (0)       \\
&NOT		&	2.724 (1.021)	&	42.490 (56.636)	&	 0.972 (0.013)	&	-0.610 (0.650)	&   3421.9  &\ (34.8)    \\
\hline
fms  
&g-SeedBS	&	0.005 (0.004)	&	15.810 (25.830)	&	 0.955 (0.037)	&	-0.020 (0.492)	&   19.1    & (0)   \\
&WBS	    &	0.004 (0.003)	&\ \ 9.080 (14.975)	&	 0.963 (0.031)	&\	 0.000 (0.284)	&   833.8   & (7.9) \\
&n-SeedBS	&	0.004 (0.003)	&	15.500 (25.853)	&	 0.958 (0.035)	&	-0.040 (0.448)	&  	19.1    & (0)   \\
& NOT		&	0.004 (0.002)	&\ \ 8.640 (15.019)	&	 0.967 (0.028)	&	-0.010 (0.266)	&   834.4   & (7.4) \\
\hline
mix  
&g-SeedBS	&	1.598 (0.517)	&	86.870 (57.740)	&	 0.908 (0.044)	&	-1.180 (1.048)	&   22.3    & (0)    \\
&WBS	    &	1.680 (0.728)	&	78.260 (58.431)	&	 0.913 (0.041)	&	-1.090 (1.093)	&   939.1   & (8.1)  \\
&n-SeedBS	&	1.759 (0.605)	&	96.870 (64.631)	&	 0.897 (0.052)	&	-1.340 (1.199)	&   22.3    & (0)    \\
& NOT		&	1.981 (0.914)	&	89.220 (64.899)	&	 0.903 (0.049)	&	-1.260 (1.495)	&   939.2   & (8.3)  \\
\hline
teeth10  
&g-SeedBS	&	0.061 (0.040)	&\ \ 7.960 (20.229)	&	 0.933 (0.130)	&	-0.190 (2.419)	&   4.4     & (0)       \\
&WBS	    &	0.058 (0.040)	&\ \ 7.380 (20.591)	&	 0.936 (0.133)	&	-0.410 (2.170)	&   238.7   & (2.2)     \\
&n-SeedBS	&	0.066 (0.050)	&	10.790 (27.521)	&	 0.911 (0.186)	&	-0.860 (2.903)	&   4.4     & (0)       \\
&NOT		&	0.062 (0.044)	&\ \ 9.150 (24.593)	&	 0.923 (0.162)	&	-0.730 (2.522)  &   238.8   & (2.3)     \\
\hline
stairs10 
&g-SeedBS	&	0.023 (0.011)	&	 2.130 (1.468)	&	 0.981 (0.013)	&\	 0.470 (0.745)	&   4.8     & (0)     \\
&WBS	    &	0.024 (0.013)	&	 1.900 (1.374)	&	 0.980 (0.014)	&\	 0.450 (0.730)	&   255.4   & (2.3)   \\
&n-SeedBS	&	0.021 (0.011)	&	 1.340 (1.249)	&	 0.984 (0.013)	&\	 0.100 (0.362)	&   4.8     & (0)     \\
& NOT	    &	0.021 (0.011)	&	 1.400 (1.326)	&	 0.984 (0.013)	&\	 0.140 (0.450)	&   255.5   & (2.5)   \\
\hline
\end{tabular}
\begin{tablenotes}
\item g-SeedBS, Seeded binary segmentation with greedy selection; 
n-SeedBS, Seeded binary segmentation with narrowest over threshold selection; 
WBS, Wild binary segmentation; 
NOT, Narrowest over threshold method; MSE, Mean squared error; Hausdorff dist., Hausdorff distance; tot. length, total length of search intervals (in thousand); signals are described in Appendix B of \citet{Fryzlewicz_WBS}.
\end{tablenotes}
\end{threeparttable}
\end{table*}

Besides increasing the computational effort by the choice of the decay $a$, there could be other options to improve empirical performance. For example, one could find preliminary change point estimates with the recommended choice of $a=1/2^{1/2}$ and include additional search intervals in a second step covering the ranges between previously found neighbouring change points. Should the best candidates from these additional intervals have test statistics above the pre-defined acceptance threshold, then one can also add them to the final set of change point estimates. Another option mentioned at the definition of seeded intervals is to limit the minimal number of observations $m$ (i.e.~consider fewer layers) if prior information on the minimal true segment lengths is available, which reduces the risk of selecting false positives coming from very small intervals during the final model selection. One could also use more refined test statistics instead of the CUSUM statistics that are adapted to various scales, e.g.~multiscale variants.

\begin{figure}[t]
\begin{center}
\includegraphics[width=0.7\textwidth]{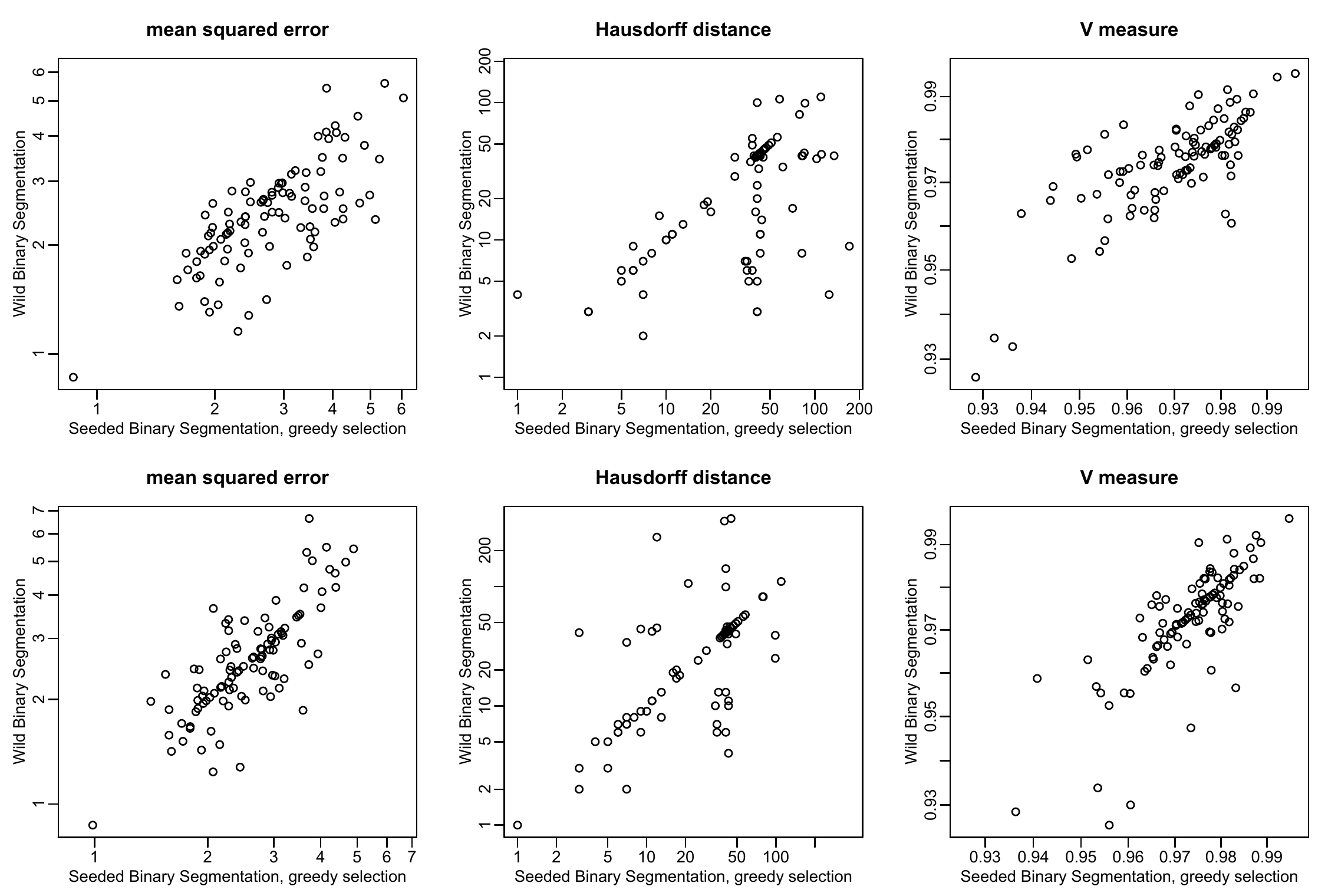}
\caption{Simulation results from Table \ref{tab:simulation_results} for the blocks signal with decay $a=1/2^{1/2}$ (top) and decay $a=1/2^{1/8}$ (bottom).}
\label{fig:estimation_visualization}
\end{center}
\end{figure}

\section{Beyond the univariate case: An example for high-dimensional models}
\label{highdim_example}
As mentioned earlier, another appealing applicability of seeded binary segmentation beyond the extensively discussed univariate Gaussian setup is in scenarios where model fits are expensive, for example in the recently emerging area of high-dimensional change point detection (e.g.~means in \cite{Soh_Chandrasekaran_mean, Wang_highdim_mean_change}; covariance and precision matrices/graphical models in \cite{RoyMichailidis, GibberdNelson, Wang_Yu_Rinaldo_theory, Gibberd_Roy, Bybee_Atchade_JMLR, Avanesov_theory, Londschien}; linear regression in \cite{LeonBuhl, Wang_Willet_regression}). Several of these proposals require the calculation of numerous lasso \citep{lasso_original}, graphical lasso \citep{glasso} or similar rather costly estimators. Another example would be multivariate non-parametric change point detection for which \cite{RF_change_point} use random forests \citep{Breiman2001} and other classifiers, \cite{Padilla_Yu_Wang_Rinaldo_multivariate_nonparametric} use kernel density estimation, while \cite{ecp} utilise energy distances. Overall, in all approaches requiring costly single fits, it is essential to keep the total length of search intervals as small as possible to make computations feasible and thus to have practically usable methods. 

We would like to highlight this in a specific example involving change point detection in (high-dimensional) Gaussian graphical models. We consider a setup and estimator as introduced by \cite{Londschien}. Let $(X_i)_{i = 1}^T \subseteq \mathbb{R}^p$ be a sequence of independent Gaussian random vectors with means $\mu_i \in \mathbb{R}^p$ and covariance matrices $\Sigma_i \in \mathbb{R}^{p \times p}$ 
such that $(\mu_i, \Sigma_i), i = 1, \ldots, T$ is piecewise constant with an unknown number $N$ of change points at unknown locations and corresponding segment boundaries $1=\eta_0 < \eta_1 < \dots < \eta_N < \eta_{N+1}=T$. For $0 \leq u < v \leq T$ let $X_{(u,v]} \in \mathbb{R}^{p \times (v-u)}$ denote the matrix of the observations $X_{u+1},\ldots,X_{v}$, denote with $\hat\mu^{(u,v]}$ their mean and let $S_{(u,v]} = (X_{(u,v]} - \hat\mu^{(u,v]}) (X_{(u,v]} - \hat\mu^{(u,v]})^t/(v - u)$ be the corresponding covariance matrix. Let $0 < \delta < 1/2$ be some required minimal relative segment length and $\lambda>0$ be a regularisation parameter. According to the proposal of \cite{Londschien}, for a segment $(u,v]$ with $v-u>2\delta T$ define the split point candidate as
\begin{equation}
\label{glasso_change_est}
\hat\eta_{(u,v]} = \argmax\limits_{s \in \{ u+\delta T,\ldots, v-\delta T\} } 
L_T(\hat\Omega_{(u, v]}^{\mathrm{glasso}}; S_{(u, v]}) - \left(L_T(\hat\Omega_{(u, s]}^{\mathrm{glasso}}; S_{(u,s]}) + 
L_T(\hat\Omega_{(s, v]}^{\mathrm{glasso}}; S_{(s, v]}) \right),
\end{equation}
where $L_T$ corresponds to a multivariate Gaussian log-likelihood based loss in the considered segment, i.e.   
\begin{equation*}
L_T(\Omega; S_{(u,v]}) =
\frac{v-u}{T}\left(\Tr(\Omega S_{(u,v]}) - \log(|\Omega|) \right),
\end{equation*}
for an arbitrary segment $(u,v]$. Moreover, $\hat\Omega_{(u,v]}^{\mathrm{glasso}}$ is the graphical lasso precision matrix estimator \citep{glasso} for the corresponding segment with a specifically scaled regularisation parameter, i.e.
\begin{equation*}
\hat\Omega_{(u,v]}^{\mathrm{glasso}} = \argmin\limits_{\mathbb{R}^{p \times p} \ni \Omega \succ 0} L_T(\Omega; S_{(u,v]}) + \sqrt{(T / (v- u))}\lambda \|\Omega\|_1.
    \label{eq:glasso_estimator}
\end{equation*}

The goal here is to compare the estimation performance and computational times of the estimator~\eqref{glasso_change_est} when applied to seeded as well as random intervals. We consider a specific chain network model \cite[Example 4.1]{fan2009_chain_network} with $p=20$ variables: $\Sigma_{ij} = \exp{(-a \abs{s_i-s_j})}$ with $a=1/2$ and $s_i - s_{i-1} = 0.75, i = 2,\ldots,p$. The resulting precision matrix $\Sigma^{-1}$ is tridiagonal. We consider a modified version $\tilde \Sigma$ by replacing the top left $5\times5$ part of $\Sigma$ by $I_5$, i.e. an identity matrix of dimension $5$. Let $F_1 = \mathcal{N}(0_p, \Sigma)$ and $F_2 = \mathcal{N}(0_p, \tilde \Sigma)$. We take $20$ segments of length $100$ each (i.e.,~$N=19$) and generate independent observations in the following way. The first hundred observations are drawn from $F_1$, the following hundred from $F_2$, then another hundred from $F_1, F_2, F_1$ etc. In total, we obtain $T=2000$ observations. We simulated $10$ different data sets this way and to each of these applied seeded binary segmentation (greedy selection) with decays $a = 1/2, 1/2^{1/2}, 1/2^{1/4}, 1/2^{1/8}, 1/2^{1/16}$ as well as $M = 100, 200, 400, 800, 1600$ random intervals for wild binary segmentation, leading to the $5$ red, respectively $5$ black point clouds in Fig.~\ref{fig:glasso_example}. To be precise, we set $\delta=0.015$, hence, we left out $30$ observations at the boundaries of background intervals. This also implies that all intervals with less than $60$ observations were discarded for both interval construction types. We set $\lambda = 0.007$ and a threshold of $0.04$ as a necessary improvement in \eqref{glasso_change_est} in order to keep the proposed split point $\hat\eta_{(u,v]}$ in the segment $(u,v]$ and declare it as a change point.

Note that with this specific choice of $\delta$ we do not have a truly high-dimensional scenario in any of the segments, but regularisation is still helpful especially when considering short background intervals (where the number of variables $p$ is of similar magnitude as the number of observations). The reasons why we kept the number of variables~$p$ comparably low are on the one hand the easier choice of tuning parameters, and on the other hand we would quickly face computational bottlenecks with wild binary segmentation when increasing~$p$. In this $20$-dimensional example, around $2300$ seconds are needed for wild binary segmentation to obtain an estimation accuracy that is already achieved in about $28$ seconds using seeded intervals, see top of Fig.~\ref{fig:glasso_example}. The bottom of Fig.~\ref{fig:glasso_example} shows computational gains of similar magnitude ($17$ vs. $1200$ seconds). The difference in computational times for the top and bottom of Fig.~\ref{fig:glasso_example} arise from the different implementation. The top shows a naive implementation where for each split point $s\in (u,v]$ both the input covariance matrices ($S_{(u,s]}, S_{(s, v]}$) as well as the corresponding graphical lasso fits ($\hat\Omega_{(u,s]}^{glasso}, \hat\Omega_{(s, v]}^{glasso}$) are computed from scratch. However, the input covariance matrices can be updated cheaply for neighbouring split points, i.e.~once $S_{(u,s]}$ is available, with a cost of $O(p^2)$ it can be updated to obtain $S_{(u,s+1]}$. This better implementation, shown at the bottom part of Fig.~\ref{fig:glasso_example}, saves $40-50\%$ in computational time both for wild and seeded binary segmentation, but still preserves the massive computational gains of the seeded intervals. The remaining computational costs mainly come from the expensive graphical lasso fits. Without computing the graphical lasso fits, in our \textsf{R} implementation one would only need about $2$ seconds for the fastest version of seeded binary segmentation ($a=1/2$) and about $9$ seconds for $M=100$ and $150$ seconds for $M=1600$ for wild binary segmentation. However, as mentioned previously, without the regularised fits of the graphical lasso, the estimation error would be clearly worse, and hence, this is not an option. Thus, the next best guess to speed up would be to use warm starts from neighbouring splits when updating the graphical lasso fits, similar to the covariance matrix updates described previously. Unfortunately, not all algorithms that have been developed to compute graphical lasso fits are guaranteed to converge with warm starts. \cite{Mazumder2012_dpglasso} proposed the DP-GLASSO algorithm that can be started at any positive definite matrix. When using the corresponding \textbf{dpglasso} \textsf{R} package on CRAN, we obtained speedups of $30-35\%$ with warm starts, but even this way the computational times were roughly $10-20$ times worse than using the graphical lasso algorithm of \cite{glasso} implemented in the \textbf{glasso} \textsf{R} package on CRAN. The latter one we used to obtain the results of Fig.~\ref{fig:glasso_example}. Admittedly, there are a lot of factors influencing computational times (e.g.~convergence thresholds) and different algorithms/implementations might be better or worse in other scenarios with a higher number of variables~$p$ and/or regularisation parameter~$\lambda$.

It is hard to draw an overall conclusion on which combination of implementation and tricks would give the fastest results for changing Gaussian graphical models, but it is even harder to be fast without utilising the more efficient seeded interval construction for the following reason.
In each background interval evaluating \eqref{glasso_change_est} at a single split point $s$ requires the calculation of $2$ graphical lasso fits. Hence, the total computational cost (for a small $\delta$) would be roughly $O(T\log(T)p^2 + T\log(T) \cdot \mathrm{gl}(p))$ for seeded binary segmentation versus $O(TMp^2 + TM \cdot \mathrm{gl}(p))$ for wild binary segmentation, where $\mathrm{gl}(p)$ denotes the cost of computing a graphical lasso estimator for a $p$-dimensional covariance matrix. There are various algorithms computing exact or approximate solutions for the graphical lasso estimator with differing computational cost, but scalings for $\mathrm{gl}(p)$ of $O(p^3)$ (or even worse) are common once the input covariance matrices are already available and unless some special structures can be exploited (e.g.~block diagonal screening proposed by \cite{Witten2011_connected_components} as well as \cite{Mazumder2012_connected_components}). The terms 
$T\log(T)p^2$ and $TMp^2$ come from the calculation of the input covariance matrices via updates. Summarising all observations from above, for typical algorithms for graphical lasso problems, the total costs would be roughly $O(T\log(T)p^3)$ for seeded binary segmentation vs. $O(TMp^3)$ for wild binary segmentation.

As in the univariate Gaussian change in mean case, $M$ needs to be much larger than $O(\log T)$ if there are many change points. In the worst case $M$ is of order $O(T^2)$ when evaluating essentially all possible background intervals. In practice, one often does not know which choice of $M$ is sufficient. This is a problem as the performance depends very much on the choice of $M$ (see Fig.~\ref{fig:glasso_example}), especially in case there are many change points. Another advantage of seeded intervals is, that results are not depending much on the choice of the decay $a$ for the seeded interval construction compared to the very strong dependence on the choice of $M$ (the number of random intervals) for wild binary segmentation. 

Overall the aim of this example is to demonstrate that seeded binary segmentation is highly useful in scenarios beyond the univariate Gaussian change in mean case. While we precisely quantify the computational gains in the changing Gaussian graphical model scenario when using graphical lasso estimators, it is more difficult to quantify this for some other costly estimators such as the lasso \citep{lasso_original} or random forest \citep{Breiman2001} for the following reasons. As a main difference, they have a cost that not only depends on the number of variables, but also on the number of observations, moreover, possibilities to update neighbouring fits would need to be considered carefully. Even in absence of exact quantification of the computational costs in these scenarios, once can expect the substantial computational gains when using seeded intervals to be crucial to keep computations feasible when model fits are very expensive.

\begin{figure}[t]
\begin{center}
\includegraphics[width=0.8\textwidth]{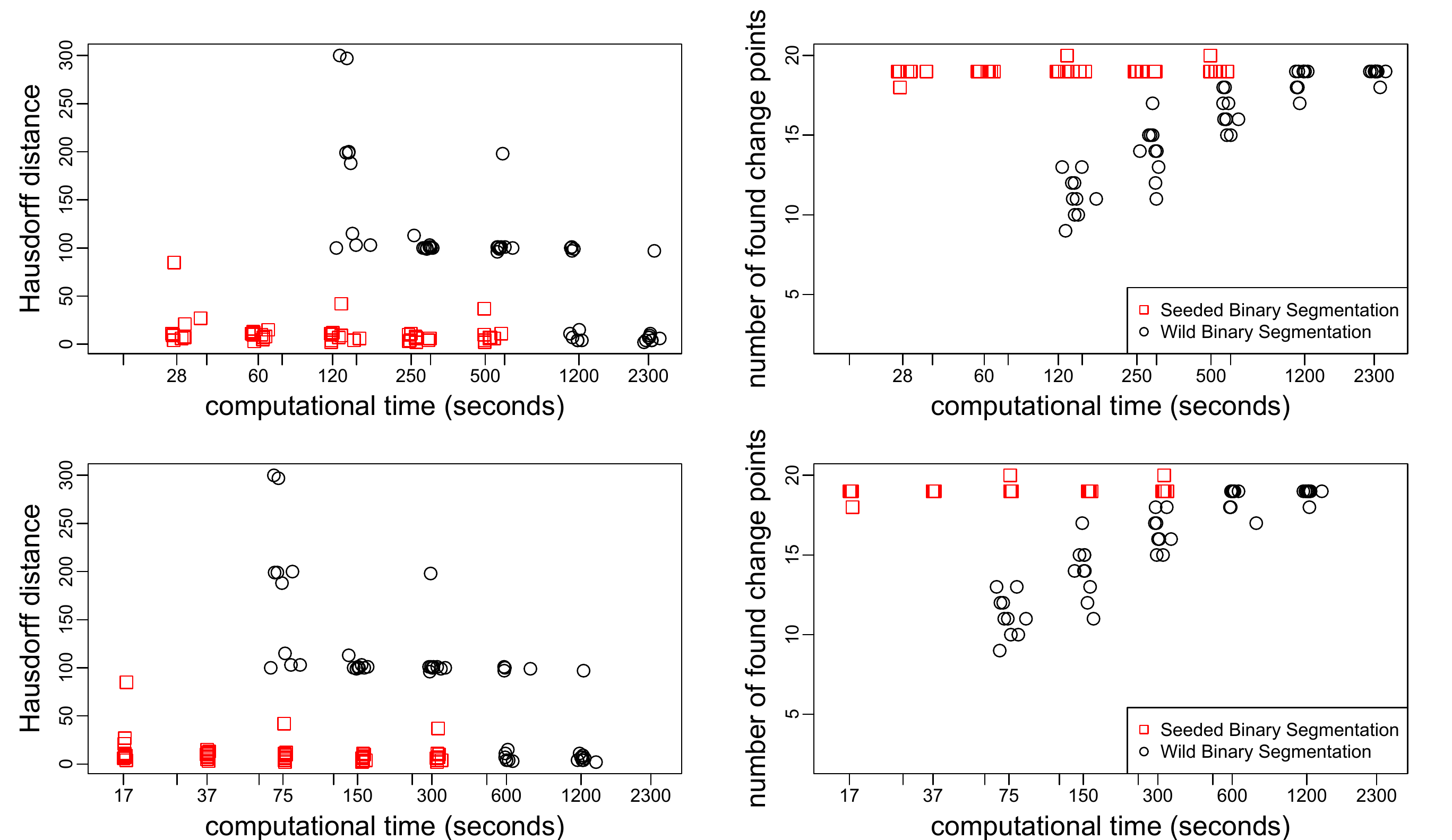}
\caption{Estimation performance on the changing Gaussian graphical model example of section \ref{highdim_example} with $2000$ observations and $19$ change points with the naive implementation (top) and the implementation using updates for neighbouring input covariance matrices (bottom). Note the logarithmic scales.}
\label{fig:glasso_example}
\end{center}
\end{figure}

\section{Conclusions}
Two popular algorithmic approaches tackling change point detection are modifications of dynamic programming and greedy procedures, mostly based on binary segmentation. While dynamic programming provides optimal partitionings, its worst case computational complexity is typically quadratic in the length of the time series. Moreover, dynamic programming often requires considerable work to be adapted to the problem at hand, in particular when aiming to do so in efficient ways, see e.g.~\cite{Killick_etal,Frick_etal,LMS16,Pein_etal,Haynes_2017_range_of_penalties, Haynes_2017_nonparametric_DP, Kernel_dynprog, Fearnhead_etal} for specific examples. Plain binary segmentation on the other hand is easy to adapt to various change point detection scenarios while being fast at the same time. The price to pay is not being optimal in terms of statistical estimation. The recently proposed wild binary segmentation \citep{Fryzlewicz_WBS} and related methods improve on statistical detection by means of considering random background intervals for the detection, but these methods lose the computational efficiency of plain binary segmentation. 

In this paper we showed that this loss of computational efficiency can be avoided by means of a systematic construction of background intervals. The so called seeded intervals have a total length of order $O(T \log T)$ compared to $O(TM)$ for wild binary segmentation with $M$ random intervals. The choice of $M$ should depend on the assumed minimal segment length which is typically unknown prior to the analysis. In frequent change point scenarios, $M$ should be clearly larger than $O(\log (T))$, in worst case $M$ is of order $O(T^2)$. The seeded interval construction is not only more efficient, but the total length of search intervals is independent of the number of change points, which even guarantees fast computations under all circumstances, unlike competing methods. Importantly, for this computational gain one does not need to sacrifice anything, i.e.~the improved statistical estimation performance remains as demonstrated in various simulation results. We have proved minimax optimality of our proposed seeded binary segmentation methodology for the univariate Gaussian change in mean problem, and demonstrated massive computational gains for detecting change points in large-dimensional Gaussian graphical models.

\section*{Acknowledgement}
Solt Kov\'acs and Peter B\"uhlmann were partially supported by the European Research Commission grant 786461 CausalStats - ERC-2017-ADG. Housen Li and Axel Munk gratefully acknowledge the support of the DFG Cluster of Excellence Multiscale Bioimaging EXC 2067.

\bibliographystyle{apalike}
\bibliography{ref}

\end{document}